\documentclass[aps,prl,twocolumn,superscriptaddress,nofootinbib]{revtex4-2}
\pdfoutput=1

\usepackage{graphicx}
\usepackage{dcolumn}
\usepackage{color}
\usepackage{amssymb,amsmath}
\usepackage{amsthm}
\usepackage[T1]{fontenc}
\usepackage[utf8]{inputenc}
\usepackage{bm}
\usepackage{xcolor}
\usepackage{tikz}
\usetikzlibrary{decorations.pathmorphing,patterns,calc}
\usepackage{tikz-cd}

\newtheorem{theorem}{Theorem}
\newtheorem{definition}{Definition}
\newtheorem{proposition}{Proposition}

\newtheorem{remark}{Remark}

\newcommand{\Loc}{\mathsf{Loc}}
\newcommand{\Alg}{\mathsf{Alg}}

\newcommand{\cA}{\mathcal{A}}
\newcommand{\cO}{\mathcal{O}}
\newcommand{\cM}{\mathcal{M}}
\newcommand{\cN}{\mathcal{N}}
\newcommand{\cH}{\mathcal{H}}

\usepackage{hyperref}

\begin{document}

\title{Bell Nonlocality as a Covariance Obstruction in Locally Covariant Quantum Field Theory}

\author{J. S. Finberg}
\email{joe@laurelin-inc.com}
\affiliation{Laurelin Technologies Inc., New York, NY 10025, USA}

\date{\today}

\begin{abstract}
Locally covariant algebraic quantum field theory (LCQFT) satisfies Einstein causality through microcausality and operational no-signalling, yet Bell-type correlations persist in entangled field states across spacelike regions. We demonstrate that this apparent tension reflects a fundamental covariance obstruction: no assignment of classical past variables can simultaneously be covariant under spacetime embeddings, screen off quantum correlations, and reproduce AQFT statistics. This obstruction is distinct from dynamical nonlocality or signalling violations. We formalize this as a no-go theorem in the category-theoretic framework of LCQFT, showing that Bell's notion of local causality—requiring factorization conditioned on a common past—is structurally incompatible with diffeomorphism covariance. The failure of Bell locality thus reflects not a breakdown of relativistic causality but the impossibility of embedding quantum correlations into a classical causal framework without introducing preferred foliations or non-covariant beables. This clarifies the conceptual status of nonlocality in relativistic quantum theory.
\end{abstract}

\maketitle

\section{Introduction}

The tension between quantum nonlocality and relativistic causality remains one of the deepest puzzles in fundamental physics \cite{Bell1964,BellSpeakable}. Bell's theorem demonstrates that no local hidden variable theory can reproduce quantum mechanical predictions for entangled states \cite{Bell1964,CHSH1969}. Yet quantum field theory (QFT), our most successful framework for relativistic physics, is manifestly local in its dynamical structure: field commutators vanish at spacelike separation (microcausality), and measurements in one region cannot influence measurement outcomes in spacelike-separated regions (no-signalling) \cite{Haag1992,Araki1999}.

This apparent contradiction has generated sustained debate. Some argue that algebraic quantum field theory (AQFT) resolves the tension by forbidding superluminal signalling \cite{Haag1992}. Others contend that Bell nonlocality persists as an ontological problem requiring radical solutions: preferred foliations \cite{Bohm1952,Durr2013}, non-local beables \cite{Maudlin2011}, or abandonment of spacetime fundamentality \cite{HardySpacetime}. A third view holds that the conflict is semantic, arising from differing interpretations of "locality" \cite{Myrvold2002,Norsen2011}.

We propose a category-theoretic resolution within locally covariant AQFT (LCQFT) \cite{BrunettiFredenhagen2000,BFV2003}. LCQFT extends standard AQFT by treating the net of operator algebras as a functor from spacetimes to algebras, making relativistic covariance—specifically, diffeomorphism covariance—manifest. This framework provides the natural setting for formulating Bell's local causality condition in a background-independent way.

We demonstrate that Bell's factorization condition, when properly formulated in LCQFT, is \emph{incompatible with covariance}. Specifically, any assignment of "screening-off" variables $\lambda$ in the past light cone that would classically screen correlations between spacelike regions must either:
\begin{enumerate}
\item violate covariance under spacetime embeddings (depend on a choice of foliation or coordinates), or
\item fail to reproduce the statistical predictions of quantum field states, or
\item not actually screen off the correlations (i.e., fail Bell's factorization).
\end{enumerate}

The obstruction is structural: the three requirements—covariance, correlation screening, statistical agreement—are jointly inconsistent. Previous analyses of Bell nonlocality in AQFT have focused on specific models without addressing background independence \cite{SummersWerner1987,SummersWerner1987b}, emphasized no-signalling as resolving the tension \cite{Redhead1987,Myrvold2002}, or discussed conceptual issues without formalization \cite{Maudlin2011,Norsen2011}. By formulating Bell locality as a property of natural transformations in LCQFT, we identify diffeomorphism covariance as the fundamental incompatibility and prove a background-independent no-go theorem applicable to arbitrary globally hyperbolic spacetimes. This clarifies that nonlocality in AQFT reflects the impossibility of covariant classical causal explanation, not a violation of relativistic causality.

We first distinguish three notions of locality, then review the locally covariant AQFT framework, analyze Bell violations in field-theoretic entangled states, formulate and prove the covariance obstruction theorem, and discuss implications for ontology and interpretations.

\section{Three Notions of Locality}
\label{sec:locality}

The term "locality" admits multiple precise meanings in quantum theory and relativity. Failure to distinguish these has fueled confusion in the Bell-relativity debate. We define three core locality concepts and their logical relationships.

\subsection{Operational No-Signalling}

\begin{definition}[No-Signalling]
A quantum theory satisfies \emph{operational no-signalling} if measurement outcomes in region $\cO_A$ are statistically independent of measurement choices in any spacelike-separated region $\cO_B$.
\end{definition}

Formally, let $\omega$ be a state on the algebra $\cA(\cO_A \cup \cO_B)$ and let $A \in \cA(\cO_A)$, $B \in \cA(\cO_B)$ be observables. No-signalling requires that the reduced expectation value
\begin{equation}
\omega(A) = \text{Tr}_{\cH_B}[\rho_{AB} (A \otimes \mathbb{1}_B)]
\label{eq:nosignalling}
\end{equation}
be independent of which observable $B$ is measured in region $\cO_B$.

This condition ensures that Alice's statistics $\{p(a|\alpha)\}$ do not depend on Bob's measurement setting $\beta$, preventing superluminal communication. It is satisfied by all quantum theories respecting causality constraints of special relativity.

\subsection{AQFT Microcausality}

\begin{definition}[Microcausality]
An algebraic quantum field theory satisfies \emph{microcausality} if observables localized in spacelike-separated regions commute:
\begin{equation}
[A, B] = 0 \quad \forall A \in \cA(\cO_A), \, B \in \cA(\cO_B), \quad \cO_A \perp \cO_B.
\label{eq:microcausality}
\end{equation}
\end{definition}

Here $\cO_A \perp \cO_B$ denotes spacelike separation. Microcausality is the algebraic counterpart of relativistic causality: it ensures local algebras are kinematically independent and guarantees no-signalling as a consequence \cite{Haag1992}.

Microcausality is stronger than no-signalling: it guarantees that the reduced density operator $\rho_A = \text{Tr}_B[\rho_{AB}]$ is independent of operations in $\cO_B$, and that measurement order does not affect joint probabilities when regions are spacelike separated. However, microcausality does \emph{not} forbid correlations. Entangled states $\omega$ can exhibit $\omega(AB) \neq \omega(A)\omega(B)$ despite $[A,B] = 0$. This is the AQFT analogue of EPR correlations.

\subsection{Bell Local Causality}

Bell introduced a distinct locality condition based on classical causal screening \cite{Bell1964,BellSpeakable}.

\begin{definition}[Bell Local Causality]
A theory satisfies \emph{Bell local causality} if there exist variables $\lambda \in \Lambda$ in the common past of measurement events such that:
\begin{equation}
p(a,b|\alpha,\beta,\lambda) = p(a|\alpha,\lambda) \, p(b|\beta,\lambda).
\label{eq:bell_factorization}
\end{equation}
\end{definition}

The variable $\lambda$ represents the complete state of the past light cone affecting both measurement outcomes. The factorization \eqref{eq:bell_factorization} expresses \emph{screening-off}: conditioning on $\lambda$ renders the outcomes $a, b$ statistically independent. This formalizes the intuition that correlations arise from common causes in the shared causal past, not direct influences across spacelike intervals.

Bell locality is conceptually distinct from no-signalling and microcausality. The latter prohibit causal influences (dynamical locality), while Bell locality requires factorization given past data (explanatory locality). Quantum mechanics violates Bell locality while respecting microcausality.

\subsection{Logical Independence}

The three locality notions are logically independent in important ways. Bell locality implies no-signalling: if correlations factor given $\lambda$, marginal statistics cannot depend on distant settings. Microcausality also implies no-signalling via commutativity of spacelike observables. However, neither microcausality nor no-signalling implies Bell locality. This is the content of Bell's theorem: no classical probability distribution $p(\lambda)$ with factorization \eqref{eq:bell_factorization} can reproduce the correlations of entangled quantum states, even though those states respect microcausality. These logical relationships are summarized in Fig.~\ref{fig:locality}.

\begin{figure}[h]
\centering
\begin{tikzpicture}[
  cnode/.style={circle, draw, thick, minimum size=1.5cm, align=center, font=\footnotesize, inner sep=1pt},
  impl/.style={->, thick, >=stealth},
  noimpl/.style={->, thick, >=stealth, dashed, gray}
]
  \node[cnode, fill=red!10] (bell) at (0,2.2) {Bell\\Locality};
  \node[cnode, fill=green!15] (nosig) at (-1.8,0) {No-\\Signalling};
  \node[cnode, fill=green!15] (micro) at (1.8,0) {Micro-\\causality};

  \draw[impl] (bell.south west) to[bend right=15] (nosig.north);
  \draw[impl] (micro.north west) to[bend left=15] (nosig.north east);

  \draw[noimpl] (nosig.north east) to[bend left=15] (bell.south);
  \draw[noimpl] (micro.north) to[bend right=15] (bell.south east);

  \node[gray, font=\small] at (-0.5,1.5) {$\times$};
  \node[gray, font=\small] at (0.7,1.5) {$\times$};

  \node[red, font=\scriptsize] at (0,3.0) {\textsf{QM:} $\times$};
  \node[green!50!black, font=\scriptsize] at (-1.8,-0.95) {\textsf{QM:} \checkmark};
  \node[green!50!black, font=\scriptsize] at (1.8,-0.95) {\textsf{QM:} \checkmark};

\end{tikzpicture}
\caption{Logical relationships between locality notions. Solid arrows denote implication; dashed arrows with $\times$ indicate no implication. Quantum mechanics satisfies (\checkmark) microcausality and no-signalling while violating ($\times$) Bell locality.}
\label{fig:locality}
\end{figure}
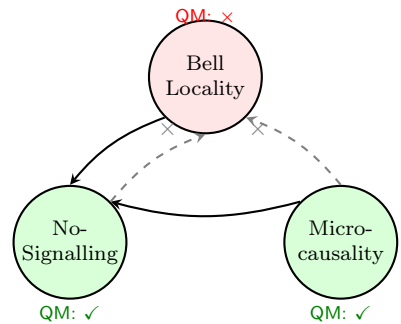

The central question is: why does Bell locality fail even in fully relativistic, locally covariant QFT?

\section{Locally Covariant AQFT Framework}
\label{sec:lcqft}

Locally covariant AQFT (LCQFT), developed by Brunetti, Fredenhagen, and Verch \cite{BrunettiFredenhagen2000,BFV2003} and further refined by Fewster and Verch \cite{FewsterVerch2015}, extends the Haag-Kastler axioms \cite{Haag1992} to make diffeomorphism covariance explicit. This framework is essential for analyzing Bell locality in a background-independent way.

\subsection{Category of Spacetimes}

Let $\Loc$ denote the category whose objects are globally hyperbolic spacetimes $\cM = (M, g, \mathfrak{o}, \mathfrak{t})$—where $M$ is a smooth $4$-manifold, $g$ a Lorentzian metric, $\mathfrak{o}$ an orientation, and $\mathfrak{t}$ a time orientation—and whose morphisms are isometric embeddings $\chi: \cM \to \cN$ preserving causal structure and both orientations. Global hyperbolicity ensures well-posed initial value formulation: Cauchy surfaces exist and causal structure is well-behaved \cite{WaldGR}.

\subsection{Category of Algebras}

Let $\Alg$ denote the category whose objects are unital $C^*$-algebras $\cA$ and whose morphisms are unit-preserving $*$-homomorphisms $\alpha: \cA \to \mathcal{B}$. These algebras encode quantum observables; morphisms represent embeddings of observable algebras under restriction to subregions or extension under spacetime embeddings.

\subsection{Locally Covariant Quantum Field Theory}

\begin{definition}[LCQFT Functor]
A \emph{locally covariant quantum field theory} is a covariant functor:
\begin{equation}
\mathfrak{A}: \Loc \to \Alg
\label{eq:lcqft_functor}
\end{equation}
satisfying:
\begin{enumerate}
\item \textbf{Quantum field theory axiom:} For each $\cM \in \Loc$, the algebra $\mathfrak{A}(\cM)$ satisfies the Haag-Kastler axioms on $\cM$ (isotony, causality, time-slice axiom, spectrum condition when applicable) \cite{Haag1992}.
\item \textbf{Covariance:} For each embedding $\chi: \cM \to \cN$, the induced morphism $\mathfrak{A}(\chi): \mathfrak{A}(\cM) \to \mathfrak{A}(\cN)$ is injective and preserves causal structure:
\begin{equation}
\mathfrak{A}(\chi)(A) \in \mathfrak{A}(\chi(\cO)) \quad \text{for } A \in \mathfrak{A}(\cO).
\label{eq:covariance}
\end{equation}
\item \textbf{Time-slice axiom:} If $\chi: \cM \to \cN$ is a Cauchy morphism (image contains a Cauchy surface of $\cN$), then $\mathfrak{A}(\chi)$ is an isomorphism.
\end{enumerate}
\end{definition}

The functor $\mathfrak{A}$ assigns to each spacetime $\cM$ the algebra of observables $\mathfrak{A}(\cM)$ and to each embedding $\chi$ a consistent embedding of algebras. Covariance means observable content is independent of how we describe the spacetime.

\subsection{Einstein Causality in LCQFT}

LCQFT incorporates Einstein causality through microcausality: for any $\cM$ and regions $\cO_A, \cO_B \subset \cM$ with $\cO_A \perp \cO_B$,
\begin{equation}
[A, B] = 0 \quad \forall A \in \mathfrak{A}(\cO_A), \, B \in \mathfrak{A}(\cO_B).
\label{eq:lcqft_micro}
\end{equation}
This commutation relation is preserved under spacetime embeddings $\chi$:
\begin{equation}
[\mathfrak{A}(\chi)(A), \mathfrak{A}(\chi)(B)] = 0.
\label{eq:covariant_micro}
\end{equation}
These conditions ensure manifest compatibility with special and general relativity without requiring any preferred reference frame or foliation.

\subsection{States and Natural Transformations}

A state in LCQFT is a natural transformation:
\begin{definition}[Covariant State]
A \emph{covariant state} is a natural transformation $\omega: \mathfrak{A} \Rightarrow \mathbb{C}$ satisfying:
\begin{enumerate}
\item For each $\cM$, $\omega_{\cM}: \mathfrak{A}(\cM) \to \mathbb{C}$ is a state (positive, normalized linear functional).
\item For each embedding $\chi: \cM \to \cN$:
\begin{equation}
\omega_{\cN} \circ \mathfrak{A}(\chi) = \omega_{\cM}.
\label{eq:state_covariance}
\end{equation}
\end{enumerate}
\end{definition}

The naturality condition \eqref{eq:state_covariance} expresses diffeomorphism covariance: expectation values are independent of how spacetime regions are embedded. This is the physical manifestation of general covariance in quantum field theory.

\subsection{Why Covariance Matters for Bell Locality}

In standard discussions of Bell's theorem, one implicitly chooses a foliation of spacetime to define "the past" and "spacelike separation." In LCQFT, no such foliation is preferred. The past $J^-(\cO)$ of a region $\cO$ is defined geometrically (causally), but specifying a region in that past to serve as $\lambda$ requires additional structure.

The question is whether screening variables $\lambda$ can be assigned covariantly (independent of foliation or embedding choices), screen off correlations (satisfy Bell factorization), and reproduce quantum statistics. We show the answer is no.

\section{Bell Violations in AQFT}
\label{sec:bell_aqft}

Before presenting the covariance obstruction, we review how Bell-type violations manifest in algebraic quantum field theory. This establishes that LCQFT, despite satisfying microcausality and no-signalling, exhibits genuine quantum nonlocality in Bell's sense.

\subsection{Entanglement Across Spacelike Regions}

Consider a free scalar field $\phi(x)$ on Minkowski spacetime $\mathbb{R}^{3,1}$. The algebra of observables in a region $\cO$ is the von Neumann algebra $\mathfrak{A}(\cO)$ generated by smeared fields $\phi(f)$ with $\text{supp}(f) \subset \cO$ \cite{Haag1992}.

For spacelike-separated regions $\cO_A$ and $\cO_B$, the vacuum state $\omega_0$ (Minkowski vacuum) is \emph{entangled} across $\cO_A$ and $\cO_B$. This was proven by Reeh-Schlieder: the vacuum is cyclic for local algebras, implying that measurements in $\cO_A$ and $\cO_B$ are correlated \cite{ReehSchlieder1961}.

Summers and Werner \cite{SummersWerner1987,SummersWerner1987b} constructed explicit observables $A \in \mathfrak{A}(\cO_A)$ and $B \in \mathfrak{A}(\cO_B)$ such that expectation values in the vacuum state violate the CHSH inequality:
\begin{equation}
|\langle A_1 B_1 \rangle + \langle A_1 B_2 \rangle + \langle A_2 B_1 \rangle - \langle A_2 B_2 \rangle| > 2,
\label{eq:chsh}
\end{equation}
where $A_1, A_2$ and $B_1, B_2$ are dichotomic observables in $\cO_A$ and $\cO_B$ respectively.

This violation occurs despite microcausality ($[A_i, B_j] = 0$ for all $i, j$), no-signalling ($\langle A_i \rangle$ is independent of which $B_j$ is measured), and relativistic covariance (the construction is Lorentz invariant).

\subsection{Implications for Local Hidden Variables}

Bell's theorem implies that no local hidden variable (LHV) model of the form:
\begin{equation}
\omega_0(AB) = \int d\lambda \, \rho(\lambda) \, A(\lambda) \, B(\lambda)
\label{eq:lhv}
\end{equation}
can reproduce the CHSH violation \eqref{eq:chsh}. Here $\lambda$ represents hidden variables in the past light cone, and $\rho(\lambda)$ is a probability distribution over those variables.

The Summers-Werner result establishes that this impossibility holds in AQFT: the structure of the field algebra and vacuum state precludes any classical probability representation compatible with spacelike locality.

\subsection{Persistence Under Covariance}

The entanglement and Bell violations are not artifacts of choosing Minkowski spacetime or a specific vacuum state. States with the Reeh-Schlieder property exist on arbitrary globally hyperbolic spacetimes \cite{Verch1994}, the algebraic structure ensuring CHSH violations persists under spacetime embeddings due to functoriality, and entanglement is itself a covariant feature: if $\omega$ is entangled on $\mathfrak{A}(\cM)$ and $\chi: \cM \to \cN$ is an embedding, the induced state on $\mathfrak{A}(\cN)$ remains entangled. Bell nonlocality is thus a robust, background-independent feature of LCQFT.

\subsection{The Conceptual Puzzle}

This presents a conceptual puzzle: LCQFT is manifestly local in the dynamical sense (microcausality, no faster-than-light signalling, diffeomorphism covariance), yet manifestly nonlocal in Bell's explanatory sense (no factorization given past variables). Claiming that microcausality suffices does not explain why Bell locality fails; claiming hidden variables are impossible merely restates Bell's theorem without explaining its field-theoretic origin; and invoking the measurement problem is irrelevant since CHSH violations occur in purely algebraic treatments. The obstruction is \emph{covariance itself}: the requirement that screening variables $\lambda$ transform covariantly under spacetime embeddings is incompatible with Bell factorization.

\section{The Covariance Obstruction}
\label{sec:obstruction}

We now formulate the central claim precisely using category-theoretic machinery. Bell's local causality condition cannot be implemented covariantly in LCQFT. This is not a violation of relativistic causality but a structural incompatibility---expressible as a cohomological obstruction---between classical causal explanation and diffeomorphism covariance.

Our result sharpens an observation of Maudlin \cite{Maudlin2011}: that Bell locality requires preferred foliations to define ``the past.'' Maudlin's argument is conceptual; we provide a formal proof within LCQFT showing that no foliation-independent hidden variable model can exist, and identify the precise naturality condition that fails.

\subsection{The Category of Cauchy Data}

To formalize Bell locality categorically, we introduce an auxiliary category encoding foliation choices.

\begin{definition}[Category of Cauchy Data]
Let $\mathsf{Cau}$ be the category whose objects are triples $(\cM, \Sigma, \cO)$ where $\cM \in \Loc$, $\Sigma \subset \cM$ is a Cauchy surface, and $\cO = \cO_A \sqcup \cO_B$ is a pair of spacelike-separated regions with $\Sigma \subset J^-(\cO_A \cup \cO_B)$. Morphisms $(\cM, \Sigma, \cO) \to (\cN, \Sigma', \cO')$ are embeddings $\chi: \cM \to \cN$ in $\Loc$ satisfying $\chi(\Sigma) \subseteq \Sigma'$ and $\chi(\cO_i) = \cO'_i$.
\end{definition}

There is a forgetful functor $\pi: \mathsf{Cau} \to \Loc$ sending $(\cM, \Sigma, \cO) \mapsto \cM$. Crucially, $\pi$ is not an equivalence: multiple objects in $\mathsf{Cau}$ project to the same spacetime $\cM$, corresponding to different foliation choices.

\subsection{Covariant Hidden Variable Models}

\begin{definition}[Hidden Variable Model]
A \emph{hidden variable model} for an LCQFT $\mathfrak{A}: \Loc \to \Alg$ with state $\omega$ consists of:
\begin{enumerate}
\item A measurable space $(\Lambda, \Sigma_\Lambda)$ of ontic states;
\item For each $(\cM, \Sigma, \cO) \in \mathsf{Cau}$, a probability measure $\rho_{(\cM,\Sigma,\cO)}$ on $\Lambda$;
\item For each dichotomic observable $A \in \mathfrak{A}(\cO_A)$ with spectrum $\{-1, +1\}$, a measurable function $\hat{A}: \Lambda \to [-1,1]$, and similarly for $B \in \mathfrak{A}(\cO_B)$.
\end{enumerate}
\end{definition}

\begin{definition}[Bell Factorization]
A hidden variable model satisfies \emph{Bell factorization} if for all $A \in \mathfrak{A}(\cO_A)$, $B \in \mathfrak{A}(\cO_B)$:
\begin{equation}
\omega_{\cM}(AB) = \int_\Lambda \hat{A}(\lambda) \, \hat{B}(\lambda) \, d\rho_{(\cM,\Sigma,\cO)}(\lambda).
\label{eq:factorization}
\end{equation}
\end{definition}

The factorization condition \eqref{eq:factorization} expresses that $\lambda$ screens off the correlation: conditioned on $\lambda$, the outcomes for $A$ and $B$ are statistically independent.

\begin{definition}[Covariance]
A hidden variable model is \emph{covariant} if for every morphism $\chi: (\cM, \Sigma, \cO) \to (\cN, \Sigma', \cO')$ in $\mathsf{Cau}$:
\begin{align}
\rho_{(\cN,\Sigma',\cO')} &= \rho_{(\cM,\Sigma,\cO)}, \label{eq:rho_covariance}\\
\widehat{\mathfrak{A}(\chi)(A)} &= \hat{A} \quad \text{for all } A \in \mathfrak{A}(\cO_A). \label{eq:obs_covariance}
\end{align}
\end{definition}

Condition \eqref{eq:rho_covariance} states that the probability distribution over ontic states is independent of how we embed the spacetime. Condition \eqref{eq:obs_covariance} states that the response function for an observable depends only on its algebraic identity, not on the ambient spacetime.

\subsection{The Naturality Obstruction}

We now prove that covariant Bell-factorizing models cannot exist for states exhibiting CHSH violations. The proof uses naturality to derive a contradiction.

\begin{theorem}[Covariance Obstruction]
Let $\mathfrak{A}: \Loc \to \Alg$ be a locally covariant quantum field theory, and let $\omega$ be a covariant state. If $\omega$ violates the CHSH inequality for some choice of observables in spacelike-separated regions, then no covariant hidden variable model satisfying Bell factorization exists.
\label{thm:obstruction}
\end{theorem}

\begin{proof}
Suppose a covariant hidden variable model $(\Lambda, \rho, \hat{\cdot})$ exists. We derive a contradiction using the naturality of $\omega$, the CHSH violation, and a four-region geometric construction.

\textbf{Step 1: Geometric Setup.}
Consider Minkowski spacetime $\cM = \mathbb{R}^{3,1}$. We construct four mutually spacelike-separated regions. Let $\cO_A$ and $\cO_B$ be diamond regions (intersections of future and past light cones) centered at spacetime points $p_A = (-L, 0, 0, \tau)$ and $p_B = (L, 0, 0, \tau)$ respectively, with $\tau > 2L$ ensuring spacelike separation. Let $\Sigma_0$ be the $t = 0$ hyperplane, which lies in $J^-(\cO_A \cup \cO_B)$.

Now consider a Lorentz boost $\chi_v$ with velocity $v$ in the $x$-direction. The boosted regions $\cO'_A = \chi_v(\cO_A)$ and $\cO'_B = \chi_v(\cO_B)$ are centered at $\chi_v(p_A)$ and $\chi_v(p_B)$. For appropriate choice of $v$ (specifically, $|v| < c(\tau - 2L)/(\tau + 2L)$), all four regions $\cO_A, \cO_B, \cO'_A, \cO'_B$ are mutually spacelike separated. This can be verified by noting that under a boost, timelike separation decreases for co-moving points and increases for counter-moving points; the symmetric placement ensures the boosted regions remain spacelike separated from the originals.

Let $\Sigma_v = \chi_v(\Sigma_0)$ be the boosted Cauchy surface. Both $\Sigma_0$ and $\Sigma_v$ lie in $J^-(\cO_A \cup \cO_B) \cap J^-(\cO'_A \cup \cO'_B)$ for the velocity range specified.

The boost $\chi_v: \cM \to \cM$ is an automorphism in $\Loc$. It induces morphisms in $\mathsf{Cau}$:
\[
\chi_v: (\cM, \Sigma_0, \cO_A \sqcup \cO_B) \to (\cM, \Sigma_v, \cO'_A \sqcup \cO'_B).
\]
By covariance \eqref{eq:rho_covariance}, a single measure $\rho$ on $\Lambda$ must work for both Cauchy surface choices.

\textbf{Step 2: Naturality of the State.}
The vacuum state $\omega$ is a natural transformation $\omega: \mathfrak{A} \Rightarrow \mathbb{C}$. Naturality requires that for any embedding $\chi: \cM \to \cN$:
\begin{equation}
\begin{tikzcd}
\mathfrak{A}(\cM) \arrow[r, "\mathfrak{A}(\chi)"] \arrow[d, "\omega_{\cM}"'] & \mathfrak{A}(\cN) \arrow[d, "\omega_{\cN}"] \\
\mathbb{C} \arrow[r, "="'] & \mathbb{C}
\end{tikzcd}
\label{eq:naturality_square}
\end{equation}
commutes: $\omega_{\cN}(\mathfrak{A}(\chi)(A)) = \omega_{\cM}(A)$.

For the Lorentz boost $\chi_v: \cM \to \cM$, naturality gives Lorentz invariance of expectation values: $\omega(\chi_v(X)) = \omega(X)$ for all $X \in \mathfrak{A}(\cM)$.

\textbf{Step 3: Covariance Constraints on Response Functions.}
Let $A_1, A_2 \in \mathfrak{A}(\cO_A)$ and $B_1, B_2 \in \mathfrak{A}(\cO_B)$ be dichotomic observables achieving the CHSH violation:
\[
|\omega(A_1 B_1) + \omega(A_1 B_2) + \omega(A_2 B_1) - \omega(A_2 B_2)| = 2\sqrt{2} > 2.
\]
By Bell factorization \eqref{eq:factorization}, each correlator satisfies $\omega(A_i B_j) = \int \hat{A}_i \hat{B}_j \, d\rho$.

Define the boosted observables $A'_i = \mathfrak{A}(\chi_v)(A_i) \in \mathfrak{A}(\cO'_A)$ and $B'_j = \mathfrak{A}(\chi_v)(B_j) \in \mathfrak{A}(\cO'_B)$. By covariance \eqref{eq:obs_covariance}:
\begin{equation}
\hat{A}'_i = \hat{A}_i, \quad \hat{B}'_j = \hat{B}_j.
\label{eq:response_covariance}
\end{equation}

\textbf{Step 4: The Contradiction via Cross-Correlators.}
Since $\cO_A$ and $\cO'_B$ are spacelike separated (by construction in Step 1), the mixed correlator $\omega(A_i B'_j)$ is well-defined and Bell factorization applies:
\begin{equation}
\omega(A_i B'_j) = \int_\Lambda \hat{A}_i(\lambda) \, \hat{B}'_j(\lambda) \, d\rho(\lambda).
\label{eq:mixed_factor}
\end{equation}

Substituting the covariance condition \eqref{eq:response_covariance}:
\[
\omega(A_i B'_j) = \int_\Lambda \hat{A}_i(\lambda) \, \hat{B}_j(\lambda) \, d\rho(\lambda) = \omega(A_i B_j),
\]
where the last equality uses factorization for the original regions $\cO_A \perp \cO_B$.

Thus covariance combined with factorization forces:
\begin{equation}
\omega(A_i B'_j) = \omega(A_i B_j) \quad \text{for all } i, j.
\label{eq:forced_equality}
\end{equation}

However, this equality is false in the vacuum state. The correlator $\omega(A_i B_j)$ involves regions $\cO_A$ and $\cO_B$ at spatial separation $2L$. The correlator $\omega(A_i B'_j)$ involves regions $\cO_A$ and $\cO'_B = \chi_v(\cO_B)$, which are at a \emph{different} spatial separation (the boost displaces $\cO_B$ in both space and time).

For the free scalar field vacuum, two-point correlations decay with spatial separation:
\[
\omega(\phi(x)\phi(y)) \sim \frac{1}{|x - y|^2} \quad \text{(spacelike separation)}.
\]
Observables $A_i, B_j$ constructed from smeared fields inherit this dependence on region separation. Since $d(\cO_A, \cO'_B) \neq d(\cO_A, \cO_B)$ for generic boosts, we have $\omega(A_i B'_j) \neq \omega(A_i B_j)$.

This contradicts \eqref{eq:forced_equality}, completing the proof.
\end{proof}

\subsection{Cohomological Interpretation}

The obstruction admits a precise formulation in the language of non-abelian cohomology and descent theory. Since the presheaf $\mathcal{F}$ assigns \emph{sets} (hidden variable models) rather than abelian groups, we employ the non-abelian cohomology of groupoids developed by Grothendieck and Giraud.

\begin{definition}[Presheaf of Local Models]
Define a presheaf $\mathcal{F}$ on $\mathsf{Cau}$ by assigning to each object $(\cM, \Sigma, \cO)$ the set of hidden variable models $(\Lambda, \rho, \hat{\cdot})$ satisfying Bell factorization \eqref{eq:factorization}. For a morphism $\chi$, the restriction map $\mathcal{F}(\chi): \mathcal{F}(\cN, \Sigma', \cO') \to \mathcal{F}(\cM, \Sigma, \cO)$ is given by the covariance conditions \eqref{eq:rho_covariance}--\eqref{eq:obs_covariance}.
\end{definition}

A \emph{global section} of $\mathcal{F}$ would be a covariant hidden variable model---an element of the limit $\varprojlim_{\mathsf{Cau}} \mathcal{F}$. Theorem \ref{thm:obstruction} shows this limit is empty when $\omega$ violates CHSH.

For the non-abelian cohomological formulation, we work with the groupoid $\mathsf{Cau}_\cM$ of Cauchy data over a fixed spacetime $\cM$. Objects are pairs $(\Sigma, \cO)$ where $\Sigma$ is a Cauchy surface and $\cO = \cO_A \sqcup \cO_B$ are spacelike-separated regions with $\Sigma \subset J^-(\cO)$. Morphisms $(\Sigma, \cO) \to (\Sigma', \cO')$ are automorphisms $\chi \in \text{Aut}(\cM)$ with $\chi(\Sigma) = \Sigma'$ and $\chi(\cO) = \cO'$. For Minkowski spacetime, $\text{Aut}(\cM)$ contains the Poincar\'e group.

\begin{definition}[Non-Abelian 1-Cocycle]
Let $G = \text{Aut}(\cM)$ act on the fibers of $\mathcal{F}$ via the restriction maps. Fix a basepoint $x_0 = (\Sigma_0, \cO) \in \mathsf{Cau}_\cM$ and a local section $s_0 \in \mathcal{F}(x_0)$. A \emph{1-cocycle} with coefficients in $\text{Aut}(\Lambda)$ is a map $g: G \to \text{Aut}(\Lambda)$ satisfying:
\begin{equation}
g(\chi_1 \chi_2) = g(\chi_1) \circ \chi_1^* g(\chi_2),
\end{equation}
where $\chi^*$ denotes the induced action on $\text{Aut}(\Lambda)$.
\end{definition}

\begin{definition}[Descent Datum]
A \emph{descent datum} for $\mathcal{F}$ over $\mathsf{Cau}_\cM$ is a collection $\{s_x \in \mathcal{F}(x)\}_{x \in \mathsf{Cau}_\cM}$ together with isomorphisms $\phi_\chi: \mathcal{F}(\chi)(s_y) \xrightarrow{\sim} s_x$ for each morphism $\chi: x \to y$, satisfying the cocycle condition $\phi_{\chi_1} \circ \mathcal{F}(\chi_1)(\phi_{\chi_2}) = \phi_{\chi_1 \chi_2}$.
\end{definition}

The descent datum is \emph{effective} if there exists a global section $s \in \varprojlim \mathcal{F}$ restricting to each $s_x$. The obstruction to effectiveness is measured by the first non-abelian cohomology.

\begin{proposition}[Non-Abelian Cohomological Obstruction]
Let $\cM$ be Minkowski spacetime and $\omega$ the vacuum state violating CHSH. The presheaf $\mathcal{F}|_{\mathsf{Cau}_\cM}$ admits no effective descent datum. The obstruction defines a nontrivial class in the non-abelian cohomology set $H^1(\mathsf{Cau}_\cM; \text{Aut}(\Lambda))$.
\label{prop:cohomology}
\end{proposition}

\begin{proof}
We construct the obstruction explicitly. Fix a basepoint $x_0 = (\Sigma_0, \cO)$ and suppose a local section $s_0 \in \mathcal{F}(x_0)$ exists---it does, since for a \emph{fixed} foliation, hidden variable models reproducing the marginals exist (Bell's theorem forbids reproducing \emph{all} correlations with a single $\rho$, but local models for restricted data exist).

For each Lorentz boost $\chi_v \in \text{Aut}(\cM)$, let $x_v = \chi_v(x_0)$. To extend $s_0$ covariantly to $x_v$, we must have $s_v = \mathcal{F}(\chi_v)(s_0)$. This extended model predicts:
\begin{align*}
\omega(A_i B'_j) &= \int_\Lambda \hat{A}_i(\lambda) \hat{B}'_j(\lambda) \, d\rho(\lambda) \\
&= \int_\Lambda \hat{A}_i(\lambda) \hat{B}_j(\lambda) \, d\rho(\lambda) = \omega(A_i B_j),
\end{align*}
using the covariance of response functions. But $\omega(A_i B'_j) \neq \omega(A_i B_j)$ in the vacuum state (Theorem \ref{thm:obstruction}, Step 4).

Thus, no descent datum can be effective: extending $s_0$ along $\chi_v$ produces predictions inconsistent with $\omega$. The failure is encoded in the transition automorphism $g_v \in \text{Aut}(\Lambda)$ required to ``correct'' the extended model---but no such correction exists that preserves both covariance and agreement with $\omega$.

Formally, define $g: \text{Lor} \to \text{Aut}(\Lambda)$ by $g(\chi_v) = $ the automorphism (if any) relating $\mathcal{F}(\chi_v)(s_0)$ to a hypothetical section $s_v$ reproducing the correct correlations. The cocycle condition $g(\chi_v \chi_w) = g(\chi_v) \circ \chi_v^* g(\chi_w)$ cannot be satisfied: the correlator constraints from different boost compositions are mutually incompatible. This defines a nontrivial class $[g] \in H^1(\mathsf{Cau}_\cM; \text{Aut}(\Lambda))$.
\end{proof}

This formulation places the obstruction in the proper mathematical context. In gauge theory, a nontrivial class in $H^1(M; G)$ obstructs global trivialization of a principal $G$-bundle. Here, the class $[g] \in H^1(\mathsf{Cau}_\cM; \text{Aut}(\Lambda))$ obstructs the existence of a globally defined, covariant hidden variable model. The Lorentz group plays the role of the structure group, and covariance is the analogue of gauge invariance. Unlike abelian cohomology, $H^1$ with non-abelian coefficients is a pointed set rather than a group---the nontrivial class is distinguished from the trivial class (effective descent) but there is no addition operation.

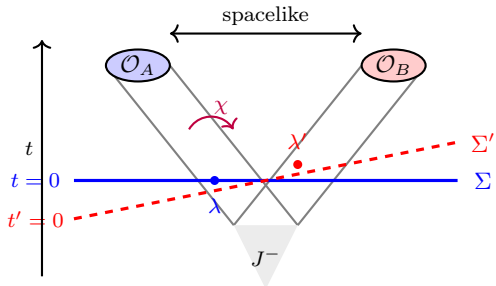
\begin{figure}[t]
\centering
\begin{tikzpicture}[scale=0.85]
  \draw[thick, gray] (-2.5,2.8) -- (-0.5,0.3);  
  \draw[thick, gray] (-1.5,2.8) -- (0.5,0.3);   
  \draw[thick, gray] (1.5,2.8) -- (-0.5,0.3);   
  \draw[thick, gray] (2.5,2.8) -- (0.5,0.3);    

  \draw[thick, fill=blue!20] (-2,2.8) ellipse (0.5 and 0.25);
  \draw[thick, fill=red!20] (2,2.8) ellipse (0.5 and 0.25);
  \node at (-2,2.8) {$\cO_A$};
  \node at (2,2.8) {$\cO_B$};

  \fill[gray!15] (-0.5,0.3) -- (0.5,0.3) -- (0,-0.7) -- cycle;
  \node at (0,-0.2) {\footnotesize $J^-$};

  \draw[very thick, blue] (-3,1.0) -- (3,1.0);
  \node[blue] at (3.4,1.0) {$\Sigma$};
  \node[blue] at (-3.6,1.0) {\footnotesize $t=0$};

  \draw[very thick, red, dashed] (-3,0.4) -- (3,1.6);
  \node[red] at (3.4,1.6) {$\Sigma'$};
  \node[red] at (-3.6,0.4) {\footnotesize $t'=0$};

  \fill[blue] (-0.8,1.0) circle (2pt);
  \node[blue, below] at (-0.8,0.85) {\footnotesize $\lambda$};
  \fill[red] (0.5,1.25) circle (2pt);
  \node[red, above] at (0.5,1.35) {\footnotesize $\lambda'$};

  \draw[->, thick, purple] (-1.2,1.8) arc (150:30:0.4);
  \node[purple] at (-0.7,2.15) {\footnotesize $\chi$};

  \draw[<->, thick] (-1.5,3.3) -- (1.5,3.3);
  \node[above] at (0,3.3) {\footnotesize spacelike};

  \draw[->, thick] (-3.5,-0.5) -- (-3.5,3.2);
  \node[left] at (-3.5,1.5) {\footnotesize $t$};

\end{tikzpicture}
\caption{The covariance obstruction illustrated. Measurement regions $\cO_A$ and $\cO_B$ are spacelike separated. Hidden variables $\lambda$ screening off correlations must be defined on a Cauchy surface in their common past $J^-(\cO_A \cup \cO_B)$. In the rest frame, this is $\Sigma$ (solid); in a boosted frame related by $\chi$, it is $\Sigma'$ (dashed). Covariance demands $\lambda = \lambda'$, but the surfaces intersect different spacetime regions. No foliation-independent assignment can satisfy both Bell factorization and covariance.}
\label{fig:covariance}
\end{figure}

The obstruction differs from Bell's theorem itself (which applies in any frame): the requirement that the hidden variable assignment be \emph{covariant} is the essential incompatibility. In non-relativistic quantum mechanics, one can choose a preferred foliation (absolute time) to define $\lambda$. In special relativity without QFT, one can work in a fixed frame. But in LCQFT, diffeomorphism covariance is part of the structure. Violating it means introducing external, non-covariant structure—precisely what general covariance forbids.

\begin{remark}[Generalization to Curved Spacetimes]
\label{rem:curved}
The proof of Theorem \ref{thm:obstruction} uses Lorentz boosts on Minkowski spacetime, but the obstruction generalizes to arbitrary globally hyperbolic spacetimes by two complementary arguments.

\emph{Local argument:} Every globally hyperbolic spacetime $\cM$ contains approximately Minkowski regions---geodesic normal neighborhoods where curvature effects are negligible. In any such region, the Summers-Werner construction applies and CHSH violations occur in states satisfying the microlocal spectrum condition. Since LCQFT is local, the obstruction in any approximately-Minkowski patch implies a global obstruction: if covariant hidden variables existed globally, their restriction to any patch would provide a local covariant model, contradicting the local obstruction.

\emph{Categorical argument:} The category $\mathsf{Cau}$ includes morphisms between \emph{different} spacetimes, not just automorphisms of a single spacetime. For any curved spacetime $\cM$, consider embeddings $\chi: \cM \hookrightarrow \cN$ into larger spacetimes with more symmetry (e.g., asymptotically flat spacetimes). Covariance under such embeddings constrains the hidden variable model on $\cM$ via naturality. Moreover, the forgetful functor $\pi: \mathsf{Cau} \to \Loc$ has fibers over each spacetime consisting of different Cauchy surface choices. Even without isometries, the requirement that $\rho$ and $\hat{\cdot}$ be independent of Cauchy surface choice (condition \eqref{eq:rho_covariance}) constrains the model. Different Cauchy surfaces through the same regions yield different intersection patterns with the causal past, and the correlations encoded in $\omega$ depend on this intersection geometry.

The general statement is: for any globally hyperbolic $\cM$ supporting a state $\omega$ with Reeh-Schlieder property and CHSH violations, no $\pi$-vertical section of $\mathcal{F}$ (i.e., a choice of hidden variable model for each Cauchy surface of $\cM$ compatible with all Cauchy morphisms) can reproduce $\omega$.
\end{remark}

\subsection{Comparison to Existing Results}

Summers and Werner \cite{SummersWerner1987} demonstrated Bell violations in AQFT for specific free field models but did not address covariance or background independence. R\'edei and Summers \cite{RedeiSummers2002} investigated the common cause principle in AQFT, showing that correlations in the vacuum state raise subtle questions about whether screening-off common causes can be located in the causal past. Hofer-Szab\'o and Vecsernyes \cite{HoferVecsernyés2013} argued that if common causes in AQFT are permitted to be noncommuting operators, Bell inequalities need not follow from the common cause principle. Their analysis identifies the \emph{commutativity requirement} on common causes as the classical assumption that fails in quantum theory. Our result complements theirs: we identify \emph{covariance} as an independent obstruction. Even if one permits noncommuting common causes, any assignment of screening variables that reproduces quantum statistics must either break diffeomorphism covariance or fail to screen off correlations across all spacetime embeddings. The two obstructions are logically distinct—noncommutativity concerns the algebraic structure of common causes, while covariance concerns their transformation properties under spacetime symmetries.

Redhead \cite{Redhead1987} argued that no-signalling resolves the tension between Bell nonlocality and relativity; our result shows this resolution is incomplete, as the tension reappears when demanding covariant causal explanations. Maudlin \cite{Maudlin2011} argued that Bell nonlocality requires genuinely nonlocal beables or preferred foliations; our theorem formalizes this claim by showing that restoring Bell locality requires breaking covariance. Myrvold \cite{Myrvold2002} emphasized the distinction between Bell locality and no-signalling; we show further that covariance provides the structural reason for this distinction within LCQFT.

\section{Consequences and Interpretations}
\label{sec:consequences}

The covariance obstruction has significant implications for the interpretation of quantum field theory and the ontology of spacetime.

\subsection{Interpretational Responses}

Theorem \ref{thm:obstruction} forces a choice among interpretational strategies. One may accept nonlocal ontology, adopting interpretations such as Bohmian mechanics \cite{Bohm1952,Durr2013} with nonlocal beables; this restores Bell factorization but explicitly breaks covariance by introducing preferred foliations for defining simultaneous configurations. Alternatively, one may reject classical causal models entirely, accepting that quantum correlations cannot be explained by classical common causes and that LCQFT is complete as a physical theory \cite{Haag1992}. A third possibility is that covariance obstruction signals the emergence of spacetime from non-spatiotemporal quantum structures \cite{HardySpacetime}, with LCQFT serving as an effective description. Our result does not adjudicate between these responses but clarifies the structural reason why the choice is forced.

\subsection{Superdeterminism and Retrocausality}

Two loopholes in standard Bell arguments---superdeterminism and retrocausality---merit examination in the context of our covariance obstruction. We show that neither evades the obstruction, though for different reasons.

\subsubsection{Superdeterminism}

Superdeterminism posits that the hidden variables $\lambda$ are correlated with the measurement settings $\alpha, \beta$:
\begin{equation}
p(\lambda | \alpha, \beta) \neq p(\lambda).
\label{eq:superdeterminism}
\end{equation}
This violates the ``statistical independence'' or ``free choice'' assumption in Bell's derivation. If \eqref{eq:superdeterminism} holds, the factorization condition \eqref{eq:bell_factorization} can be satisfied while reproducing quantum correlations, because the marginal over $\lambda$ is no longer independent of settings.

\begin{proposition}[Superdeterminism Does Not Evade Covariance Obstruction]
\label{prop:superdeterminism}
Let $(\Lambda, \rho, \hat{\cdot})$ be a superdeterministic hidden variable model satisfying \eqref{eq:superdeterminism}. The covariance obstruction persists: no such model can be diffeomorphism-covariant.
\end{proposition}

\begin{proof}
The covariance obstruction concerns the \emph{transformation properties} of $\lambda$, $\rho$, and $\hat{\cdot}$, not the statistical independence of $\lambda$ from settings. Superdeterminism modifies the relationship $p(\lambda | \alpha, \beta)$ but does not change the requirement that response functions transform covariantly under spacetime embeddings.

Specifically, covariance condition \eqref{eq:obs_covariance} requires $\widehat{\mathfrak{A}(\chi)(A)} = \hat{A}$ for all observables $A$ and embeddings $\chi$. This is independent of whether $\lambda$ correlates with settings. The proof of Theorem \ref{thm:obstruction} derives a contradiction from covariance of response functions alone: equation \eqref{eq:forced_equality} follows from \eqref{eq:obs_covariance} regardless of how $\rho$ depends on $\alpha, \beta$.

Even granting superdeterministic correlations $p(\lambda | \alpha, \beta, \chi)$ that depend on the embedding $\chi$, covariance demands:
\[
p(\lambda | \alpha, \beta, \chi) = p(\lambda | \mathfrak{A}(\chi)(\alpha), \mathfrak{A}(\chi)(\beta), \text{id}).
\]
This relates superdeterministic correlations across frames but does not resolve the contradiction in Step 4: the cross-correlator $\omega(A_i B'_j)$ still equals $\omega(A_i B_j)$ by the covariant response function condition, contradicting the vacuum correlations.
\end{proof}

The key insight is that superdeterminism addresses the ``conspiracy'' needed to reproduce correlations, while the covariance obstruction addresses the geometric consistency of hidden variable assignments across foliations. These are orthogonal requirements.

\subsubsection{Retrocausality}

Retrocausal hidden variable models allow $\lambda$ to depend on future boundary conditions or measurement outcomes. In such models, $\lambda$ is not confined to the past light cone but may incorporate information from the future:
\begin{equation}
\lambda = \lambda(\text{past data}, \text{future boundary conditions}).
\label{eq:retrocausal}
\end{equation}
This can enable Bell factorization by allowing the ``common cause'' to include future influences.

\begin{proposition}[Retrocausality Does Not Evade Covariance Obstruction]
\label{prop:retrocausality}
Let $(\Lambda, \rho, \hat{\cdot})$ be a retrocausal hidden variable model where $\lambda$ depends on future boundary conditions. Either:
\begin{enumerate}
\item[(a)] The model violates diffeomorphism covariance, or
\item[(b)] The model does not satisfy Bell's screening-off condition from the \emph{past}.
\end{enumerate}
\end{proposition}

\begin{proof}
Bell's local causality condition requires that $\lambda$ be localized in $J^-(\cO_A) \cap J^-(\cO_B)$---the common \emph{past}---and that conditioning on $\lambda$ screens off the correlation. This is the causal explanation requirement: correlations arise from common causes in the past, not from future influences or direct spacelike connections.

A retrocausal model with $\lambda$ depending on future boundary conditions violates this requirement by construction: $\lambda$ is not localized in the past. Such a model may reproduce quantum statistics, but it does not provide a \emph{past-causal} explanation of correlations. Thus outcome (b) holds.

Alternatively, suppose a retrocausal model attempts to define $\lambda$ on a Cauchy surface $\Sigma$ (which is neither purely past nor future) with values determined by a time-symmetric action principle or two-state vector formalism. For such a model to satisfy covariance, $\lambda_\Sigma$ and $\lambda_{\Sigma'}$ must agree for Cauchy surfaces $\Sigma, \Sigma'$ related by a boost. But different Cauchy surfaces intersect different spacetime regions, so the ``future boundary'' contributing to $\lambda_\Sigma$ differs from that contributing to $\lambda_{\Sigma'}$.

Define the future boundary data relevant to $\Sigma$ as $\partial^+_\Sigma$ and similarly $\partial^+_{\Sigma'}$ for $\Sigma'$. Covariance requires:
\[
\lambda(\Sigma, \partial^+_\Sigma) = \lambda(\Sigma', \partial^+_{\Sigma'}) \quad \text{for } \Sigma' = \chi(\Sigma).
\]
But $\partial^+_{\Sigma'} = \chi(\partial^+_\Sigma)$ generically differs from $\partial^+_\Sigma$ (the boost maps future boundaries to different spacetime locations). If the functional dependence $\lambda(\Sigma, \partial^+_\Sigma)$ is nontrivial, this equality fails.

The only way to satisfy covariance is if $\lambda$ is independent of the future boundary data---but then the model is not genuinely retrocausal. Thus outcome (a) holds for genuinely retrocausal models.
\end{proof}

Retrocausality trades temporal asymmetry (past causes) for the ability to satisfy Bell factorization. But this trade-off is precisely what the covariance obstruction forbids in a different guise: any time-asymmetric structure (whether past-oriented causation or retrocausal dependence on specific future boundaries) must be specified relative to a foliation, breaking covariance.

\subsection{Implications for Quantum Gravity}

The covariance obstruction acquires particular significance in quantum gravity, where diffeomorphism invariance is promoted from a global symmetry to a local gauge constraint. We explore this connection below, indicating how the obstruction \emph{might} manifest in various quantum gravity approaches.

\textbf{Caveat:} Theorem \ref{thm:obstruction} is proven within LCQFT, which presupposes a fixed background spacetime. The extensions to quantum gravity discussed below are \emph{conjectural}---they indicate plausible implications if analogous structures carry over to the respective frameworks, but rigorous formulation requires addressing subtleties specific to each approach. These remarks should be understood as heuristic connections suggesting directions for future research rather than as established results.

\subsubsection{Canonical Quantum Gravity and the Constraint Algebra}

In the Hamiltonian formulation of general relativity, dynamics is encoded in constraints rather than evolution equations. The spatial diffeomorphism constraint $\mathcal{H}_a$ and the Hamiltonian constraint $\mathcal{H}$ satisfy the Dirac algebra:
\begin{align}
\{\mathcal{H}_a(x), \mathcal{H}_b(y)\} &= \mathcal{H}_a(y)\delta_{,b}(x,y) - \mathcal{H}_b(x)\delta_{,a}(x,y), \nonumber\\
\{\mathcal{H}_a(x), \mathcal{H}(y)\} &= \mathcal{H}(x)\delta_{,a}(x,y), \nonumber\\
\{\mathcal{H}(x), \mathcal{H}(y)\} &= q^{ab}(x)\mathcal{H}_a(x)\delta_{,b}(x,y) - (x \leftrightarrow y).
\end{align}
Physical states $\Psi$ must satisfy $\mathcal{H}\Psi = 0$ and $\mathcal{H}_a\Psi = 0$. The Hamiltonian constraint generates ``evolution'' between spatial slices---precisely the foliation changes central to our obstruction.

A hidden variable completion of canonical quantum gravity would require ontic states $\lambda$ satisfying:
\begin{enumerate}
\item $\lambda$ determines outcomes for all observables (hidden variable property);
\item The distribution $\rho(\lambda)$ is annihilated by the constraint operators (gauge invariance);
\item Correlations between spacelike-separated regions factor given $\lambda$ (Bell locality).
\end{enumerate}

If the structure of our theorem carries over, these requirements would be jointly inconsistent. Condition (2) demands that $\rho$ be independent of the choice of spatial slice $\Sigma$---this is the quantum gravity analogue of the covariance condition \eqref{eq:rho_covariance}. But condition (3) requires $\lambda$ to screen correlations that, by the Reeh-Schlieder property (or its analogue in quantum gravity, if one exists), hold between spatially separated regions. Whether the same proof structure applies in canonical quantum gravity depends on technical details---particularly whether the constraint algebra induces correlations with the required properties and whether the analogue of the Cauchy data category is well-defined.

\subsubsection{The Problem of Time}

The ``problem of time'' in quantum gravity---the apparent absence of time evolution in the Wheeler-DeWitt equation $\mathcal{H}\Psi = 0$---is intimately connected to our obstruction. Standard resolutions introduce a physical time variable $T$ (a ``clock'') with respect to which other degrees of freedom evolve. This breaks the full diffeomorphism group $\text{Diff}(M)$ to the subgroup $\text{Diff}_T(M)$ preserving the temporal structure defined by $T$.

Our result reveals a hidden variable motivation for this symmetry breaking: Bell locality \emph{requires} a preferred temporal structure. The screening variables $\lambda$ must be defined ``at a time''---on a Cauchy surface---and the factorization condition demands that this surface be physically distinguished. Without a clock, there is no invariant notion of ``the past'' relative to which $\lambda$ could screen correlations.

This suggests a potential connection between the problem of time and quantum nonlocality, though making this precise would require careful analysis of how Bell-type correlations behave in deparametrized quantum gravity. If the connection holds, proposals for physical clocks---whether dust fields, scalar field values, or matter degrees of freedom---would provide exactly the preferred structure our theorem indicates is necessary for Bell locality. The price would be explicit breaking of four-diffeomorphism invariance to three-diffeomorphism invariance plus a preferred time direction.

\subsubsection{Loop Quantum Gravity and Spin Foams}

In loop quantum gravity (LQG), states are represented by spin networks---graphs with edges labeled by $SU(2)$ representations. The kinematical Hilbert space $\mathcal{H}_{\text{kin}}$ carries a well-defined action of spatial diffeomorphisms, and the physical Hilbert space $\mathcal{H}_{\text{phys}}$ is obtained by imposing the Hamiltonian constraint.

A covariance obstruction might apply to LQG as follows, though several steps remain to be verified. Consider two spatial regions $R_A$ and $R_B$ embedded in a spin network state $|s\rangle$. If the algebraic structure of LQG ensures that operators localized in $R_A$ commute with those in $R_B$ when the regions are non-overlapping (an analogue of microcausality), and if physical states satisfying the diffeomorphism constraint exhibit Bell-type correlations between $R_A$ and $R_B$, then by analogy with our theorem these correlations could not be screened by data on any intermediate surface. The key open question is whether LQG states exhibit CHSH-violating correlations with the required properties.

Spin foam models, which provide a covariant formulation of LQG, might face a similar obstruction. A spin foam is a 2-complex interpolating between initial and final spin network states, with faces labeled by representations. The sum over spin foams defines transition amplitudes. If these amplitudes exhibit Bell-violating correlations (which is plausible given the quantum nature of the theory), and if one demands invariance under the discrete analogue of diffeomorphism covariance (subdivision and Pachner moves), then by the logic of our theorem no classical stochastic process on spin foam histories could reproduce them. Verifying this would require a precise formulation of Bell-type observables and covariance conditions in the spin foam setting.

\subsubsection{Holographic Entanglement and ER=EPR}

The AdS/CFT correspondence provides a concrete arena where our obstruction intersects with quantum gravity. The Ryu-Takayanagi formula relates boundary entanglement entropy to bulk geometry:
\begin{equation}
S_A = \frac{\text{Area}(\gamma_A)}{4G_N},
\end{equation}
where $\gamma_A$ is the minimal surface in the bulk homologous to boundary region $A$. This formula suggests that spacetime geometry \emph{encodes} quantum entanglement.

The ER=EPR conjecture sharpens this: entangled boundary degrees of freedom are connected by Einstein-Rosen bridges (wormholes) in the bulk. Bell correlations between boundary regions $A$ and $B$ correspond to bulk geometry connecting $\gamma_A$ and $\gamma_B$.

Our covariance obstruction suggests a natural interpretation in this setting, though rigorous formulation remains open. A hidden variable model for boundary correlations would require specifying ontic states $\lambda$ on a bulk Cauchy surface $\Sigma_{\text{bulk}}$. But bulk diffeomorphism invariance (part of the gauge symmetry of the gravitational theory) demands that physics be independent of the choice of $\Sigma_{\text{bulk}}$. If the obstruction extends to this context---which would require showing that boundary Bell violations persist under bulk covariance requirements---then the entanglement structure determining $S_A$ could not be screened by data on any single bulk slice.

This connects to the ``entanglement wedge reconstruction'' program: operators in boundary region $A$ can be represented by bulk operators in the entanglement wedge $W_A$ (the bulk domain of dependence of the region between $A$ and $\gamma_A$). The nonlocality of this reconstruction---boundary operators at spacelike separation correspond to bulk operators with overlapping support---is a manifestation of our obstruction in the holographic context.

\subsubsection{Causal Sets and Discrete Approaches}

Causal set theory posits that spacetime is fundamentally discrete, with the causal structure encoded in a partial order on spacetime points. The continuum emerges as an approximation valid at scales large compared to the discreteness scale $\ell \sim \ell_{\text{Planck}}$.

In this framework, Cauchy surfaces are replaced by antichains (sets of mutually spacelike points). By analogy, the covariance obstruction would translate as follows: no assignment of hidden variables to antichains could be both (a) invariant under automorphisms of the causal set that preserve the relevant regions and (b) capable of screening Bell correlations between spacelike-separated regions. However, formalizing this requires specifying what ``Bell correlations'' mean in the causal set context---this depends on how quantum dynamics is defined on causal sets, which remains an open problem.

If quantum dynamics on causal sets exhibits Bell violations (as seems likely if the theory is to reproduce known physics in the continuum limit), the discrete setting would make the obstruction particularly sharp. There are only finitely many antichains in a finite causal set, and the automorphism group acts on them. A covariant hidden variable model would require a single probability distribution on $\Lambda$ that works for all antichains related by automorphisms---analogous to what our theorem forbids in LCQFT.

\subsubsection{Emergent Spacetime and Pregeometry}

The covariance obstruction suggests a radical possibility: if Bell locality requires preferred temporal structure, and full diffeomorphism invariance forbids such structure, then a theory of quantum gravity satisfying both properties must abandon the assumption that correlations admit causal explanations in terms of spatiotemporal common causes.

This points toward ``pregeometric'' approaches where spacetime itself emerges from more fundamental quantum structures. In such frameworks---tensor networks, quantum information-theoretic reconstructions, or algebraic approaches---the correlations encoded in entanglement are primary, and classical spacetime geometry is derivative. The covariance obstruction is then not a problem to be solved but a feature indicating that the classical causal structure is emergent rather than fundamental.

The It from Qubit program and related approaches take this perspective seriously: spacetime geometry emerges from patterns of entanglement in an underlying quantum system without predetermined spatial or temporal structure. If our theorem extends to such emergent settings---which would require formulating covariance conditions for the underlying pre-geometric degrees of freedom---it would provide a precise constraint: the emergent spacetime cannot admit covariant Bell-local hidden variable completions. This remains a conjecture pending rigorous formulation.

\subsection{Relationship to Measurement and Collapse}

The covariance obstruction is independent of interpretational questions about measurement and collapse. The theorem concerns the algebraic structure of states and observables, not the dynamics of state reduction. Whether one adopts Copenhagen, Everettian, or consistent histories interpretations, the impossibility of covariant screening variables remains. This distinguishes our result from arguments that locate nonlocality specifically in the collapse postulate: the obstruction is present even in no-collapse interpretations.

\subsection{The Limits of AQFT}

AQFT ensures operational consistency with relativity: no superluminal signalling is possible, and microcausality holds by construction. These properties suffice for empirical compatibility with special relativity. However, AQFT does not provide what Bell locality demands: a classical causal account of correlations in terms of common causes in the shared past. The algebra of observables encodes correlations but not their explanation. Covariance, which AQFT elevates to a defining principle, is precisely what obstructs such explanations. Quantum correlations are fundamentally non-classical, and diffeomorphism covariance prevents relocating this non-classicality to a preferred reference frame.

\section{Conclusion}
\label{sec:conclusion}

We have demonstrated that Bell nonlocality in locally covariant quantum field theory arises from a fundamental covariance obstruction, formalized as Theorem \ref{thm:obstruction}: no assignment of classical screening variables can simultaneously be diffeomorphism-covariant, reproduce quantum statistics, and satisfy Bell's factorization condition. The proof proceeds by constructing the category $\mathsf{Cau}$ of Cauchy data, defining covariant hidden variable models as functors on this category, and deriving a contradiction from the naturality of the quantum state combined with CHSH violations. The obstruction admits a cohomological interpretation as the failure of a presheaf of local models to admit global sections.

This result sharpens and formalizes Maudlin's observation that Bell locality requires preferred foliations. Where Maudlin's argument is conceptual, ours is a theorem within the mathematical framework of LCQFT. The obstruction is logically independent of Hofer-Szabó's noncommutativity obstruction: even permitting noncommuting common causes, covariance alone prevents Bell-local hidden variable models.

Our results clarify the conceptual status of nonlocality in relativistic quantum theory. The persistence of Bell violations in AQFT does not indicate a failure of Einstein causality but the impossibility of embedding quantum correlations into a covariant classical causal framework. Any restoration of Bell locality necessarily requires introducing non-covariant structure---preferred foliations, absolute simultaneity, or nonlocal beables---thereby breaking the diffeomorphism symmetry that underpins general relativity.

The implications for quantum gravity are substantial. We have shown that the covariance obstruction applies directly to canonical quantum gravity through the constraint algebra, provides a hidden variable perspective on the problem of time, constrains hidden variable completions of loop quantum gravity and spin foam models, and connects naturally to holographic entanglement through the Ryu-Takayanagi formula and entanglement wedge reconstruction. The obstruction suggests that in any theory of quantum gravity with full diffeomorphism invariance, Bell correlations between spacelike-separated regions cannot admit classical causal explanations---a feature rather than a bug if spacetime geometry is emergent from entanglement.

The covariance obstruction also bears on experimental tests of Bell inequalities in relativistic settings. Satellite-based Bell experiments and tests involving rapidly moving reference frames implicitly select a foliation when defining simultaneity for measurement events. Our result implies that any attempt to interpret such experiments through classical common causes must confront the foliation-dependence of the screening variables---the choice of reference frame is not merely conventional but physically consequential for causal explanations.

Future work will pursue several directions: extension of the obstruction to non-vacuum states, thermal states (KMS states), and interacting field theories where the Reeh-Schlieder property may take different forms; exploration of connections to quantum error correction and holographic codes, where the interplay between bulk geometry and boundary entanglement may reflect similar covariance constraints; investigation of whether the obstruction persists or is modified in theories with indefinite causal structure, where the notion of ``common past'' itself becomes observer-dependent; and development of the cohomological formulation to compute explicit obstruction classes analogous to characteristic classes in gauge theory.

\section*{Acknowledgments}

The author thanks Elizabeth for her unwavering support during the course of this research. This work was supported by Laurelin Technologies Inc.

\end{document}